\documentclass[10pt]{article} 

\usepackage[utf8]{inputenc} 
\usepackage{geometry} 
\usepackage{graphicx} 
\usepackage{booktabs} 
\usepackage{array} 
\usepackage{paralist} 
\usepackage{verbatim} 
\usepackage{subfig} 
\usepackage{fancyhdr} 
\usepackage{sectsty}
\allsectionsfont{\sffamily\mdseries\upshape} 
\usepackage[nottoc,notlof,notlot]{tocbibind} 
\usepackage[titles,subfigure]{tocloft} 

\def\powertower#1#2{#1\ifnum#2>1 ^{\powertower{#1}{\numexpr#2-1\relax}}\fi}

\usepackage{amsmath}
\usepackage{amssymb}
\usepackage{amsmath}
\usepackage{microtype}
\usepackage{graphicx} 
\usepackage{tikz}
\usetikzlibrary{arrows,positioning} 
\tikzset{
    >=stealth',
    punkt/.style={
           rectangle,
           rounded corners,
           draw=black, very thick,
           text width=6.5em,
           minimum height=2em,
           text centered},
    pil/.style={
           ->,
           thick,
           shorten <=2pt,
           shorten >=2pt,}
}
\tikzset{edge/.style = {->,> = latex'}}
\usetikzlibrary{shapes.misc}

\usepackage{url}
\usepackage[colorlinks=true]{hyperref}
\usepackage{algorithm}
\usepackage{algorithmic}
\usepackage{multicol}
\usepackage{ragged2e}

\newtheorem{thm}{Theorem}[section]
\newtheorem{lm}[thm]{Lemma}

\newtheorem{qq}{Problem}
\newtheorem{eg}{Example}
\newtheorem{defi}{Definition}

\newenvironment{proof}[1][Proof]{\begin{trivlist}
\item[\hskip \labelsep {\bfseries #1}]}{\end{trivlist}}
\newenvironment{definition}[1][Definition]{\begin{trivlist}
\item[\hskip \labelsep {\bfseries #1}]}{\end{trivlist}}

\newcommand{\qed}{\nobreak \ifvmode \relax \else
      \ifdim\lastskip<1.5em \hskip-\lastskip
      \hskip1.5em plus0em minus0.5em \fi \nobreak
      \vrule height0.75em width0.5em depth0.25em\fi}

\newcount\Comments  
\newcommand{\kibitz}[2]{\ifnum\Comments=1{\color{#1}{#2}}\fi}

\newcommand{\omt}[1]{}

\DeclareMathOperator{\cut}{cut}
\DeclareMathOperator{\eval}{eval}
\renewcommand{\(}{\left(}
\renewcommand{\)}{\right)}

\begin{document}

\title{Fair Division via Social Comparison} 
\author{Rediet Abebe, Jon Kleinberg, and David C. Parkes}
\date{}

\maketitle
\begin{abstract}

  \justify 

In the classical problem of cake cutting (also known as
fair division), a resource must be divided among agents with
different utilities so that each agent believes they have
received a fair share of the resource relative to the other agents.
We introduce a variant of the problem in which there is a graph
on the agents modeling an underlying social network, and agents
only evaluate their shares relative to their neighbors' in the network.
This formulation captures many situations in which it is unrealistic
to assume a global view by the agents, and we also find that
it exposes interesting phenomena in the original problem.

Specifically, we say that an
allocation is {\em locally envy-free} if no agent envies a
neighbor's allocation and {\em locally proportional} if each agent
values her own allocation as much as the average value of her
neighbor's allocations, with the former implying the latter. While
global envy-freeness implies local envy-freeness, global
proportionality does not imply local proportionaity, or vice versa.
A general result is that for any two distinct graphs on the same set
of nodes and an allocation, there exists a set of valuation
functions such that the allocation is locally proportional on one
but not the other.

%

We fully characterize the set of graphs for which an oblivious
single-cutter protocol--- a protocol that uses a single agent to cut
the cake into pieces ---admits a bounded protocol for locally
envy-free allocations in the Robertson-Webb model, and we give a
protocol with $O(n^2)$ query complexity.  
We also consider the {\em price of envy-freeness}, which compares
the total utility of an optimal allocation to the best utility of an allocation 
that is envy-free.
We show that a lower bound 
of $\Omega(\sqrt{n})$ on the price of envy-freeness for global
allocations~\cite{ckkk} in fact holds for local
envy-freeness in any connected undirected graph.  In this sense,
sparse graphs surprisingly do not provide more flexibility with
respect to the quality of envy-free allocations.
\end{abstract}

\justify

\section{Introduction}\label{intro}

The fair allocation of resources is a fundamental problem 
for interacting collections of agents.
A central issue in fair allocation is the process by which
each agent compares her allotment to those of others'.
While theoretical models have tended to focus on global comparisons ---
in which an agent makes comparisons to the full population ---
a rich line of empirical work with its origins in the social sciences has
suggested that in practice, individuals often focus their comparisons
on their social network neighbors.
This literature, known as \emph{social comparison theory},
dates back to work of Festinger \cite{scp}, and has
been explored extensively by economists and sociologists
since; for example, see Akerlof \cite{sdsd} and Burt \cite{burt}.
The primary argument is that in many contexts, an individual's
view of their subjective well-being is based on a comparison 
with peers, defined through an underlying social network structure,
rather than through comparison with the overall population
\cite{mcbride}.

In this work, we find that the perspective of social comparison
theory motivates a rich set of novel theoretical questions in classical 
resource allocation problems.
%
In particular, we consider the \emph{cake cutting} problem, which
refers to the challenge of allocating a single divisible, continuous,
 good in a fair and efficient manner. The ``cake''  is intended
to stand for a good over which different agents have different
preferences for difference pieces. 
This problem has a
wide range of applications including international border settlements,
divorce and inheritance settlements, and allocating shared
computational resources. 

\omt{A historically prominent instance took
place at the Potsdam conference, held after the defeat of Germany in
1945. One of the goals of the conference was to divide up Berlin and
Austria into four occupation zones \cite{btbook}.}
%
%

Agent preferences are modeled through functions that map subintervals
of the $[0, 1]$ interval, which represents the entire cake, to real
numbers according to the value the agent assigns to that piece.
%
We  normalize these
valuations so that each agent's value for the whole cake is $1$.
The entire cake is to be allocated, but agents need not 
receive a single, continguous interval (and valuations are
additive across pieces).  

Following our goal of understanding the properties of 
local comparisons to network neighbors,
we study cake cutting in a setting where there is an
underlying network on the agents, and fairness considerations are
defined locally to an agent's neighbors' in the network.
Given a graph $G$ and a cake
$[0, 1]$ to be allocated, we define a \emph{locally proportional
  allocation} to be one where each agent values her allocation 
at least as much as the average value from the allocations given
to her neighbors in $G$. We define
a \emph{locally envy-free allocation} to be one where no agent envies
the allocation of any neighbor in $G$.
Analogous to graphical games~\cite{kearnsAGTbook}, it seems quite plausible
that agents may care about fairness within a local part of the
population rather than with respect to the population as a whole.

As in the global case, it is straightforward to see that a locally
envy-free allocation for $G$ is also locally proportional for $G$.
It is also clear that if $H$ is a subgraph of $G$ on the same node set
(i.e. if it contains a subset of the edges of $G$), then a
locally envy-free allocation for $G$ must also be locally envy-free for $H$,
since the constraints defining local envy-freeness for $G$ are a
superset of the constraints defining local envy-freeness for $H$.
For local proportionality, however, the constraints for different graphs
$G$ and $H$ (even when one is a subgraph of the other) can operate
quite differently, and so as a first question we ask:

\begin{qq}\label{p0}
  For graphs $G$ and $H$, 
  what is the relationship between the set of locally proportional
  allocations for $G$ and the set of locally proportional
  allocations for $H$?
\end{qq}


Network topology plays a crucial role in our results. Note that
if the network under consideration is the complete graph $K_n$, then the
local definitions coincide with their global analogues; in this sense,
the local formulations contain the standard definitions as special
cases.  At the other extreme, if the network is the empty graph $I_n$,
then any allocation satisfies local envy-freeness and proportionality.
In light of this, we can pose the following problem:

\begin{qq}\label{p1}
  Are there non-trivial classes of graphs for which we can give
  efficient protocols for locally envy-free or at least locally
  proportional allocations?
\end{qq}

We believe that, in addition to capturing real-world contexts,
posing the cake cutting problem on a network will give further insight
into the structure of the original problem.
%
%
%

We also consider the effect of  fairness on welfare as
measured through the \emph{price of fairness}~\cite{ckkk}.
The price of fairness is defined as the worst case ratio over all
inputs between the social welfare of the optimal allocation (the
allocation maximizing the sum of agent valuations) and the social
welfare of the optimal fair allocation--- the envy-free or proportional
allocation that maximizes the sum of agent valuations.
We will refer to this ratio in the case of  envy-free allocations
as the {\em price of envy-freeness}, and the ratio in the case
of proportional allocations as the {\em price of proportionality}.
When there is an underlying graph $G$ governing the comparisons,
these ratios become the {\em price of local envy-freeness} and the
{\em price of local proportionality} for $G$. We pose the following question:
\begin{qq}\label{p2}
  How do the achievable lower and upper bounds on the price of 
  local envy-freeness and local proportionality 
  depend on the structure of the graph $G$?
\end{qq}

Caragiannis et al.~\cite{ckkk} give an $\Omega(\sqrt{n})$ lower bound
for the price of global envy-freeness and proportionality, a matching
upper bound of $O(\sqrt{n})$ for the price of global proportionality,
and a loose upper bound of $n - 1/2$ for the price of global
envy-freeness.  \smallskip

\paragraph*{\bf Overview of Results.}
With respect to Problem~\ref{p0}, we show that in fact the set of
locally proportional allocations do not satisfy any natural containment
relations.  In particular, for any two distinct connected graphs 
$G$ and $H$ on the same set of nodes, there exists a set of valuations
for the agents and an allocation that is locally proportional for $G$
but not for $H$.  Note that this includes the case where $G$ is the
complete graph, and so global proportionality does not imply 
local proportionality on any other connected graph $H$.

For Problem~\ref{p1}, we start from the structure of the classical 
Cut-and-Choose solution for two agents: one agent divides the cake and
the other selects a piece.  This solution does not provide any
guarantees for global allocations with more than two agents, but
with other underlying graphs $G$ it turns out to have a natural
and non-trivial generalization.
Specifically, we fully characterize the family of graphs $G$ for
which a locally envy-free allocation can be produced a by a
protocol in which a single designated node performs all the cuts
at the outset, based only on its valuation function.

Finally, for Problem~\ref{p2}, we start from the $\Omega(\sqrt{n})$ lower
bound on the price of global envy-freeness, which 
provides a lower bound for the price of
local envy-freeness in complete graphs.  We show that this
$\Omega(\sqrt{n})$ lower bound on the price of local envy-freeness
holds in any connected, undirected graph $G$.  We consider this to be
surprising, because one might think that sparse graphs provide more
flexibility with respect to the quality of envy-free allocations.
The known upper bound for the price of global envy-freeness
serves as a loose upper bound for the price of local proportionality
and the price of local envy-freeness.  \medskip

\omt{These include directed
graphs for which all cycles pass through a common node (we refer to
these as {\em cones} of directed acyclic graphs), directed graphs in
which where each node is included in at most one cycle, and directed
pseudoforests.  We also show how these local definitions allow for
protocols for certain proper subgraphs of $K_4$ that are much simpler
than the recent Aziz-Mackenzie protocol (which is designed for the
hard case of $K_4$).  Finally, we present a case--- a subclass of
generalized windmill graphs ---in which we are able to give a locally
proportional protocol, but do not know a corresponding bounded locally
envy-free protocol.}

In outline, Section \ref{prelim} defines local envy-freeness and local
proportionality, and shows that the former implies the latter, and
that both are implied by global envy-freeness.  This section also
establishes a lack of equivalence between global proportionality and
either of the two local fairness concepts. 
Section~\ref{efna} considers the notion of protocols with a single cutter
as described above, and it characterizes the family of graphs
(formally, cones of directed acyclic graphs and their subgraphs)
for which such a protocol can yield a locally envy-free allocation.
Section~\ref{pof} turns to the price of fairness. We conclude by
giving a list of open directions that we hope Problems \ref{p0}, \ref{p1}, 
and \ref{p2} will inspire.

\subsection{Background on the Cake Cutting Problem}\label{related}

Cake cutting algorithms can be traced back to the Hebrew Bible. In the
Book of Genesis, when Abraham and Lot decided to separate, they were
presented with the challenge of dividing up the land. Abraham
suggested to mark the place on which they stand as the cutting point
and asks Lot to pick a side. Lot chose the side that is well-watered
and Abraham went in the opposite direction. This is precisely the 
Cut-and-Choose Protocol from above, which has been shown to give 
envy-free allocations for two agents. 

The cake cutting problem has a
wide range of applications including international border settlements,
divorce and inheritance settlements, and allocating shared
computational resources. 
The formal
study of the cake cutting problem  was
initiated in the 1940s due to Banach et al.~\cite{btbook,kna,ste}. 
%
Later, Steinhaus
observed that the Cut-and-Choose protocol could be extended to three
players, and asked whether it can be generalized for any number of
agents~\cite{btbook}.  This was resolved affirmatively for
proportional allocations by Banach et al. \cite{ste} and
initiated an interesting line of research for envy-free allocations;
see Brams and Taylor~\cite{btbook}. 
See Procaccia~\cite{pro2,pro1} for a recent survey from a computer
science perspective.

A central problem has been finding envy-free protocols. In 1988,
Garfunkel~\cite{gar} even called this among the most important
problems in 20th century mathematics. An envy-free protocol for any
number of agents was proposed by Brams and Taylor~\cite{bt}, but may
need an unbounded number of queries even for four agents. They thus
posed the question whether there are bounded, envy-free protocols for
$n \geq 4$. Aziz and Mackenize~\cite{am} recently provided a bounded
protocol for $n = 4$. The case for $n > 4$ remained open until recently 
when Aziz and Mackenize announced a discrete and bounded protocol~\cite{am2}. 
Although bounded, their solution has a very high multiple-exponential query 
complexity and it remains an open question whether there are more 
efficient protocols.

The query model was formalized by Roberton and Webb~\cite{rw}. 
The focus in this model has been to minimize the number of
$\cut$ and $\eval$ queries. The $\cut$ query, formally $\cut(x, y, \alpha)$, asks 
an agent for a cut point $x$ such that the subinterval $[x, y]$ has value $\alpha$ 
to that agent and the $\eval$ query, formally $\eval(x, y)$ asks the agent for 
her ver valuation of the subinterval $[x, y]$. It 
is possible to encode all (decentralized) cake cutting protocols using
this model. Note that these queries do not result in an allocation themselves 
but rather provide information about the agents' valuations in order 
to compute an allocation. The query complexity of a cake cutting 
problem is thus the worst case number of queries required in order to output an 
allocation satisfying the desired fairness criteria. 

Procaccia~\cite{pro3} gives a lower bound of $\Omega(n^2)$ for
envy-free cake cutting, which shows a gap between this and the
$\Theta(n \log n)$ bound known for the proportional cake cutting
problem~\cite{edp}.
%
%
Special cases of the problem have also received attention.
Kurokawa et al.~\cite{klp} for example
establish that an algorithm that computes
an envy-free allocation for $n$ agents with piecewise uniform
functions using $f(n)$ queries would also
be able to compute an envy-free
allocation with general valuation functions using at most $f(n)$
queries.
%
Considering only
contiguous allocations,
Stromquist~\cite{stro} shows that for any $n \geq 3$, there are no
finite envy-free cake cutting algorithms,
even with  an unbounded number of queries. 
A positive result has been achieved by considering approximately
envy-free allocations \cite{pro2}. An allocation is said to be
$\epsilon$-envy-free if each agent values the allocation of any other
agent to be at most $\epsilon$ more than their own. This relaxation
provides a simple algorithm for any $n$ using $n \lceil1 \slash
\epsilon \rceil$ cut queries, as shown in \cite{pro1}.
%
Alon~\cite{al} and Dubins  and Spanier~\cite{dp} also 
study allocations under 
alternate notions of fairness.
%
%

Due to the difficult nature of the envy-free cake cutting problem,
researchers have imposed restrictions on different aspects in order to
give useful protocols and gain insight into the problem. A few
examples are: restricting valuation functions to be only piecewise
constant or uniform~\cite{klp}, relaxing envy-freeness to
approximate envy-freeness \cite{pro1}, considering partial allocations
that simultaneously satisfy envy-freeness and proportionality
\cite{tjc}, and limiting allocations to be contiguous pieces \cite{stro}.

Incentive compatibility is not a standard consideration
in the cake cutting literature. One important
exception is the work of 
Chen et al.~\cite{tjc}, who  give a polynomial time algorithm outside of
the Robertson-Webb model for finding proportional and
envy-free allocations for  piecewise uniform valuation functions
and
any $n$, while
also achieving strategyproofness.
%
%
%

To the best of our knowledge, there is no preexisting work on the
cake cutting problem in which fairness is determined via comparisons
defined by an underlying graph.
In perhaps the most closely related paper, Chevaleyre et al. \cite{cem}
analyze how the network topology of negotiation
affects the convergence to an envy-free allocation for
multiple indivisible goods. 
In their setting, agents are only able to negotiate with, and also 
envy agents that are in their neighborhood. They ask under what 
conditions a sequence of negotiations can lead to a state where there
is no envy.  There are a number of differences between this work and ours:
first, they consider indivisible goods, which leads to a different
set of questions; and second, in their setting the network constrains
not just the comparisons that are made in determining fairness, but
also the allowable interactions in the division protocol.

\section{Relating Global and Local Properties}\label{prelim}

Let $N=\{1, 2, \cdots, n \}$ denote the set of agents.  
The cake is represented using the interval $[0, 1]$ and a \emph{piece
  of cake} is a finite union of non-overlapping (interior disjoint)
subintervals of $[0, 1]$. Allocated pieces are a finite union of
subintervals.
Each agent $i$ has a valuation function $V_i$ that maps subintervals
to values in $\mathbb{R}$. Given subinterval $[x, y] \subseteq [0,
1]$, we write $V_i(x, y)$ instead of $V_i([x, y])$ for simplicity. We
assume that valuation functions are additive, non-atomic, and
non-negative.  Non-atomicity gives us $V_i(x, x) = 0$ for all $x
\in [0, 1]$, so we can ignore boundaries
when defining cut-points. 
We  normalize  valuations so that 
$V_i(0, 1) = 1$ for each agent $i$. 
%
%

\begin{defi}[Allocation]
  An \emph{allocation} is a partition of the $[0,1]$ interval into $n$ pieces
  $\{A_1, A_2, \ldots, A_n\}$ such that $\cup_i A_i = [0, 1]$ and the
  pieces are pairwise disjoint. Each agent $i$ is assigned the
  corresponding piece $A_i$.
\end{defi}

As is standard, this ensures that the entire cake is allocated. If we
remove this constraint, then we can have trivial solutions that
satisfy fairness, such as assigning each agent nothing in the case of
envy-freeness.  This assumption is a natural one to make since the
valuation functions are assumed to be non-negative and additive 
and thus satisfy free-disposal. A \emph{feasible allocation} is one 
where no subinterval is assigned to more than one agent.  

In order to avoid direct revelation of valuation functions, which may be 
cumbersome, cake cutting procedures are typically given as protocols 
that interact with the agents in order to output a feasible allocation. 
The \emph{Robertson-Webb query model}, which is typically used for 
cake cutting protocols is defined with the following two types of queries:
\begin{itemize}
\item $\eval_i(x, y)$; this asks agent $i$ for the valuation $V_i(x, y)$.
\item $\cut_i(x, y, \alpha)$: given $y, \alpha \in [0, 1]$, this asks agent 
$i$ to pick $x \in [0, 1]$ such that $V_i(x, y) = \alpha$. 
\end{itemize}

These queries are used to gather information regarding the valuations
of the agents and need not directly determine an allocation. Rather, a
cake cutting protocol can use other steps for determining
allocations. 
%
%
%
%
\begin{defi}[Query complexity]
  The \emph{query complexity} of a cake cutting protocol is the worst
  case number of queries that the protocol requires to output an
 allocation over all possible valuation
  functions.
\end{defi}

The query complexity of a cake cutting problem is the minimum query
complexity over all known protocols for computing the desired
allocation.

\subsection{Global and Local Fairness}

Given a set of agents and an allocation $\mathcal{A} = (A_1, A_2,
\cdots, A_n)$, we formally define two global fairness criteria:
\begin{defi}[Proportional, Envy-free]
An allocation $\mathcal{A}$ is  \emph{proportional} if $V_i(A_i) \geq 1/n$, for all $i \in N$, and is \emph{envy-free} if $V_i(A_i) \geq V_i (A_j)$, for all $i, j \in N$. 
\end{defi}

Suppose we are given a directed graph $G = (V, E)$, where the nodes
correspond to agents and edges signify relations between the agents.
In particular, we assume that given a directed edge $(i, j)$, agent
$i$ can view agent $j$'s allocation. Agent $i$'s
{\em neighborhood} is the set of all nodes to which it has directed
edges $(i,j)$, and we denote this set of nodes by $N_i$.
We define $i$'s {\em degree} to be $d_i = |N_i|$.
We  define local
analogues for  fairness concepts:
\begin{defi}[Local proportional, local envy free]
\label{lf}
Given a graph $G$, an allocation $\mathcal{A}$ is \emph{locally
  proportional} if $V_i (A_i) \geq \dfrac{\sum_{j \in N_i}
  V_i(A_j)}{|N_i|}$ for all $i$ and $j \in N_i$
and
\emph{locally envy-free} if $V_i(A_i) \geq V_i(A_j)$.
\end{defi}

In a locally proportional allocation, each agent assigns as much value
to her allocation as the average value she has for a neighbors'
allocation. In a locally envy-free allocation, each agent values her
allocation at least as much as her neighbors' allocation.

When $G = K_n$, the complete graph on $n$ vertices, these local
fairness definitions coincide with their global analogues. Whereas, 
if $G= I_n$, the empty graph on $n$ nodes, then any allocation
is trivially locally envy-free. So, the graph topology plays a significant 
role in computing locally fair allocations. 
\begin{lm}
  A locally envy-free allocation $\mathcal{A}$ on some graph $G$ is
  also locally envy-free on all subgraphs $G' \subseteq G$.
\end{lm}
\begin{proof}
  We want to show that given a node $u$ and $v \in N_u$, $u$ does not
  envy $v$'s allocation in $G'$. This follows from the fact that
  $\mathcal{A}$ is a locally envy-free allocation on $G$, and if $(u,
  v)$ is an edge in $G'$, then it is also an edge in $G$.
\end{proof}

One consequence of this lemma is that local envy-freeness is implied by
global envy-freeness. 
Since globally envy-free allocations exist for all sets of
agent valuations \cite{al}, 
a locally envy-free allocation exists for every graph $G$ and every
set of agent valuations.  

\begin{lm}
  If an allocation $\mathcal{A}$ is locally envy-free on a graph
  $G$, then it is also locally proportional on the same graph.
\end{lm}
\begin{proof}
  If an allocation $\mathcal{A} = (A_1, A_2, \ldots, A_n)$ is locally
  envy-free, then for any $i \in V$, $V_i(A_i) \geq V_i(A_j), \forall
  j \in N_i$. Therefore, $V_i (A_i) \geq (\sum_{j \in N_i}
    V_i(A_j))/|N_i|$.
\end{proof}

Therefore, locally proportional allocations also exist. By
considering $G = K_n$, we also recover that global envy-freeness
implies global proportionality.
While global envy-freeness implies local envy-freeness, global
proportionality does not necessarily imply local proportionality, or
vice versa, the former of which violates intuition. We provide a counter example.
\begin{eg}
  Let $n = 4$ and $G = C_4$, the cycle graph on 4 nodes, where the nodes are labeled clockwise. 
Assume agents 2, 3, and 4
  have the uniform valuation function $V_i(x, y) = |y - x|$ for any subinterval
  $(x, y) \subseteq [0, 1]$. Let agent 1 have the piecewise uniform
  valuation function where $V_1\({0, 1/4}\) = 1/2$ and $V_1\({ 3/4, 1} \) = 1/2$, and no value for the remaining subinterval. It is easy to verify that
  the following allocation is locally proportional on $K_4$,
$$\mathcal{A} = \({[0, 1 \slash 8), [1 \slash 8, 3 \slash 8), [3 \slash 8, 5 \slash 8), [5 \slash 8, 1]}\).$$ 
In particular, $V_i(A_i) = 1 \slash 4$ 
for $i \in \{1, 2, 3\}$ and $V_4(A_4) = 3 \slash 8$. 
This allocation is however not locally proportional on $C_4$ 
since $V_1(A_1 \cup A_2 \cup A_4) = 1$, but
$V_1(0, 1 \slash 8) = 1 \slash 4 < 1/3$. It is also not locally envy-free since $V_1(A_4) > V_1(A_1)$. 
\label{ex:5}
\end{eg}

\omt{
We also show that a locally proportional allocation on $G \neq K_n$ is
not necessarily globally proportional. 
\begin{eg}
  We again consider $G = C_4$, and agents 2, 3, and 4 with the same
  valuation functions as in Example~\ref{ex:5}.
 Assume that agent 1 has the piecewise uniform
  function where $V_1(0, 1 \slash 2) = 1$, and $V_1 (1
  \slash 2, 1) = 0$. The following allocation is a locally proportional
allocation on $C_4$:
$$\mathcal{A} = \( { [0, 1 \slash 10), [4 \slash 10, 7 \slash 10), [1 \slash 10, 4 \slash 10), [7 \slash 10, 1]} \).$$ 
It is easy to check that this allocation is locally proportional, but not 
globally so. 
\end{eg}
}

We prove a stronger result regarding any pair of distinct graphs. Note by $N_i(H)$ we mean agent $i$'s neighborhood set in graph $H$. 

\begin{thm}
  Given any pair of distinct, connected graphs $G, H$ on the same
  set of nodes, there exists a valuation profile of the agents and an 
  allocation $\mathcal{A}$ such that $\mathcal{A}$ 
  is locally proportional on $G$ but not on $H$. 
\end{thm}
\begin{proof}
  First, consider the case where $H$ is a strict subgraph of $G$. Pick a node
  $i$ such that $|N_i(H)| < |N_i(G)|$. Let $N_i(G) = \{i_1, i_2, \cdots, i_k\}$ and 
  $N_i(H) = \{i_1, i_2, \cdots, i_\ell\}$ for some $\ell < k$.
  Assume that all other nodes besides $i$ have a uniform valuation
  function over the entire cake. Then, the allocation 
$A_j = \({ (j-1)/n, j/n }\)$ is locally 
proportional from the perspective of every other agent $j \in N$ on 
both $H$ and $G$. Now, define $i$'s valuation function to be the 
piecewise uniform valuation function where, 
$V_i\((i-1)/n, i/n\) = 1/(|N_i(G)| + 1)$ and
  $V_i(A_{i_1} \cup A_{i_2} \cdots \cup A_{i_\ell}) = 1 -
  {1/(|N_i(G)| + 1)}$. Agent $i$'s valuation for the allocation of nodes
$\{i_{\ell+1}, i_{\ell+ 2}, \cdots, i_{k}\}$ as well as $V(G) \backslash N_i(G)$ is $0$. 
This allocation $\mathcal{A}$ is therefore locally proportional on $G$. 
For $\mathcal{A}$ to be locally
  proportional on $H$, we need $V_i(A_i) \geq 1 \slash |N_i(H)|$. However, 
since $|N_i(H)| < |N_i(G)|$, and $V_i(A_i \cup A_{i_1} \cup A_{i_2} \cup \cdots \cup A_{i_{\ell}} )= 1$, 
we only have that $V_i(A_i) < 1/(|N_i(H)| + 1)$.

  Now, suppose $H \nsubseteq G$. Then, there exists an edge $(i, j)$ in
  the edge-set of $H$ that is not in the edge-set of $G$. Assume that
  all nodes $k \neq i$ have a uniform valuation over the entire cake. 
  Suppose further that we have the allocation where each $k$ is
assigned the piece $A_k = \({ (k-1)/n, k/n }\)$. As above, 
this allocation is locally proportional from the perspective of each 
agent $k$ on both $G$ and $H$. Define $i$'s valuation function
to be $V_i\({ (j-1)/n, j/n}\) = 1$ and $0$ on the 
remainder of the cake. Then,
  $V_i(A_i) = 0$ and $V_i (A_k) = 0$ for all $k \neq j$. Since $j \notin N_i(G)$,
 this allocation is locally 
proportional on $G$. However, it is not locally proportional on $H$
since $V_i\( { \cup_{\ell \in N_i}A_{\ell} } \) = 1$, but $V_i(A_i) = 0$. 
\end{proof}

\begin{figure}[t]
\centering
\begin{tikzpicture}[node distance=1cm, auto,]
 \node[punkt, inner sep=5pt] (GEF) {\small{Global\\ Envy-freeness}};
 \node[punkt, inner sep=5pt,below=1cm of GEF] (GP) {\small{Global\\ Proportionality}};
 \node[punkt, inner sep=5pt, right= 2cm of GEF] (LEF) {\small{Local\\ Envy-freeness}};
 \node[punkt, inner sep=5pt,below=1cm of LEF] (LP) {\small{Local\\ Proportionality}};

    \path [->] (GEF) edge[thick] (LEF);   
    \path [->] (LEF) edge[thick] (LP);   
    \path [->] (GEF) edge[thick] (GP);   

\end{tikzpicture}
\caption{Relationship Between Fairness Concepts}

\end{figure}
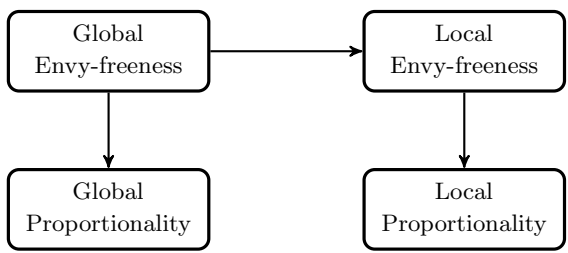

\section{Envy-Free Network Allocations}
\label{efna}

In this section, we consider the question of finding efficient protocols 
for computing locally envy-free allocations. We assume that graphs $G$ 
are directed, unless specified otherwise. When we mean the 
{\em component} 
of a directed graph, we will instead take the graph obtained by 
replacing each directed edge with an undirected one, and a component 
in the directed graph is the corresponding subgraph to the connected 
component in the undirected analogue. We will use {\em strongly
  connected component} when we mean to take directed reachability into
account. 

\begin{lm}\label{hgg'}
Suppose we have a bounded protocol for computing locally envy-free
allocations on $G$. The same protocol can be used to compute locally
envy-free allocations on the following two classes of graphs: 
(i) $H = G \cup G'$ where $G$ and $G'$ are disjoint components, 
and (ii) when $H$ is a
graph with a directed cut such that every edge across the cut goes from a node
in $G$ to a node in $H \backslash G = G'$.  
\end{lm}

\begin{proof}
In both instances, we simply apply the protocol on $G$ and allocate
agents in $G'$ the empty allocation. This is a locally envy-free
allocation on $H$, since no two agents in $G$ envy one another by
the assumption on the protocol, and no agent envies the allocation 
of another agent in $G'$. 
\end{proof}

A few consequences of this lemma are that: 
given a graph $G$ with more than one connected component, 
we can reduce the search for a bounded protocol on $G$ to any one 
of the connected components. Furthermore, 
if $G$ is a \emph{directed acyclic graph} (DAG), then there exists 
at least one node with no incoming edges. Therefore, 
the allocation where such a node gets the entire cake--- or where it is
divided among a set of such nodes ---is locally envy-free.

\subsection{Directed Acyclic Graphs and Their Cones}

We consider a conceptually useful class of graphs for which we can give 
a protocol with query complexity of $O(n^2)$. 

\begin{definition} 
  Given a graph $G = (V, E)$, we say that $G'$ is a \emph{cone} of $G$
  if it is the join of $G$ and a single node $c$, which we call the
  {\em apex}. That is, $G'$ has node set $V \cup \{c\}$, and edge set
  consisting of the edges of $G$, together with undirected edges $(u,
  c)$ for all $u \in V$. We denote the cone $G'$ of $G$ by $G \star c$.
\end{definition}


We consider cones of DAGs. These are the class of graphs where there 
is a single node $c$ that lies on all cycles. 
We now show how to compute a locally envy-free allocation on any
graph that is the cone of a DAG.
\begin{figure}[h!]
\centering
\begin{tikzpicture}
  [scale=0.75,auto=left,every node/.style={circle,fill=gray!20}]
  \node (n1) at (3, 2) {c};
  \node (n2) at (0,0)  {2};
  \node (n3) at (1.5, 0)  {3};
  \node (n4) at (4.5, 0) {n-1};
  \node (n5) at (6, 0)  {n};

  \foreach \from/\to in {n1/n2,n1/n3,n1/n4,n1/n5}
	\draw (\from) -- (\to);
 	\draw[thick, dashed] (n3) -- (n4);
	\path [->] (n2) edge[bend right=30] (n3);
	\path [->] (n2) edge[bend right=35] (n4);
	\path [->] (n2) edge[bend right=45] (n5);
	\path [->] (n3) edge[bend right=35] (n4);
	\path [->] (n3) edge[bend right=45] (n5);
	\path [->] (n4) edge[bend right=30] (n5);

\end{tikzpicture}
\caption{Cone of a Directed Acyclic Graph}
\end{figure}
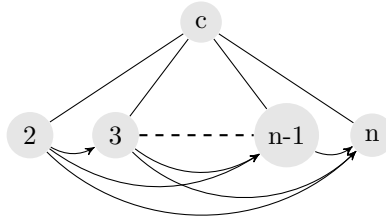

\begin{algorithm}[H]
\floatname{algorithm}{Protocol 1:}
\renewcommand{\thealgorithm}{}  
\caption{Cone of DAGs}
\begin{algorithmic}[1]
\STATE Agent $c$ cuts the cake into $n$ pieces that she values equally. 
\STATE Topologically sort and label the nodes $N \backslash \{c\}$ such that for every edge $(i, j)$, $i \leq j$. 
\STATE Nodes $N \backslash \{c\}$ pick a piece they prefer most in increasing order of their index. 
\STATE Agent $c$ takes the remaining piece.
\end{algorithmic}  
\end{algorithm}

\begin{thm}\label{thm:cdag}
Given a graph $G$ that is a cone of a DAG, Protocol 1 computes a locally envy-free allocation on $G$ using a bounded number of queries. 
\end{thm}
\begin{proof}
We first show that the allocation is locally envy-free. First, there is no envy between agent $c$ and any other agent since agent $c$ cuts the cake into $n$ pieces she values equally. Therefore, her valuation for all the allocated pieces is $1 \slash n$. Each of the other agents picks a piece before $c$, and so are able to pick a piece that they value at least as much as the remaining piece that agent $c$ is assigned. Finally, given any directed edge $(i, j)$ such that $i, j \neq c$, note that $V_i(A_i) \geq V_i (A_j)$, since if such an edge exists, then $i < j$ and thus $i$ selects a piece before $j$. 

To count the number of queries, the first step requires $n-1$ $\cut$ queries by agent $c$. Then, each agent $i$ must perform $n - i + 2$ $\eval$ queries to determine the piece for which they have the highest value. Therefore, the protocol above uses $(n^2 + 3n - 4)/2$ queries. 
\end{proof}


The importance of cones of DAGs can be seen in the following result, 
which shows that they emerge naturally as the characterization of graphs
on which a particular fundamental kind of protocol succeeds.
\begin{defi}
An {\em oblivious single-cutter protocol}
is one in which a single agent $i$ first divides up the cake into a set of pieces
$P_1, P_2, \ldots, P_t$ (potentially $t>n$), and then all remaining operations consist of agents
choosing from among these pieces.
\label{def:single-cutter}
\end{defi}

The classical Cut-and-Choose protocol is an oblivious single-cutter
protocol that works for all sets of valuation functions on the
complete,
two-node graph $K_2$.
Protocol 1
is an oblivious single-cutter protocol
which works for any graph that is the cone of a DAG.  We show that
subgraphs of cones of DAGs are in fact precisely the graphs on which
oblivious single-cutter protocols are guaranteed to produce a locally
envy-free allocation.
%

\begin{thm}
If $G$ is a graph for which an oblivious single-cutter protocol produces
a locally envy-free allocation for all sets of valuation functions,
then $G$ is a subgraph of the cone of a DAG.
\label{thm:single-cutter}
\end{thm}
\begin{proof}
Suppose, by way of contradiction, that $G$ is not a subgraph of the cone
of a DAG, but that there is an oblivious single-cutter protocol on $G$,
in which a node $i$ starts by dividing the cake into pieces using
only knowledge of her own valuation function.
Since $G$ is not a subgraph of the cone of a DAG, the graph
$G_i = G \backslash \{i\}$ is not acyclic, so there is a cycle
$C = (c_1, c_2, \cdots, c_m, c_1)$ in $G_v$.

Let $i$ have a valuation function such that she produces a partition
of the cake into pieces $P_1, P_2, \ldots, P_t$.
Because agent $i$ produces these pieces without knowledge of the valuation
functions of the other agents, we can imagine that we adversarially
choose the valuations of the other agents after these pieces have been
produced.
In particular, consider the valuation functions in which each node
$c_j$ on the cycle $C$ values piece $P_r$ (for $r < t$) at $2 \cdot 3^{-r}$,
and the last piece $P_t$ with the remaining value.
These valuations have the property that for each $r$, the piece $P_r$
is more valuable than the union of all pieces 
$P_{r+1} \cup P_{r+2} \cup \cdots \cup P_t$.

After the protocol is run, each agent $c_j$ on $C$ will get a subset
of the pieces produced by $i$.
Let $s$ be the minimum index of any piece allocated to an agent $c_j$ on $C$.
Then, the agent $c_{j-1}$ who has a directed edge to $c_j$ will have
a union of pieces that she values less than she would value $P_s$, 
and hence envies $c_j$.
This contradicts the assumption that the protocol produces a
locally envy-free allocation on $G$.
\end{proof}

We highlight an important connection between computing 
locally envy-free allocations
on graphs and what is known in the literature as {\em irrevocable 
advantage}. Given a partial allocation, an agent $i$ is said 
to have irrevocable advantage over agent $j$
(or {\em dominate} $j$) if agent $i$ remains unenvious of 
agent $j$'s allocation even if
the entire remaining piece of the cake (the {\em residue}) 
is added to agent $j$'s allocation. 

With the additional 
guarantee that each agent dominates some number of other 
agents, this concept is often 
used to extend partial globally envy-free allocations to complete
ones. For instance, 
it is a key concept in the Aziz-Mackenize protocol for $K_4$. 
Their protocol can be decomposed to three subprotocols: 
Core, Permutation, and Post-Double Domination Protocols, in order. 
The {\em Core Protocol} computes partial 
envy-free allocation where each agent dominates 
at least 
two other agents, while the {\em Post Double Domination Protocol} extends this to a complete allocation.
We will use Protocol 1
to show that given a partial envy-free 
allocation on $K_n$ where each agent dominates 
at least $n -2$ other agents, we can extend the allocation to 
a complete one, thereby generalizing the Post Double Domination 
Protocol~\cite{am} for any $n$. This is presented in Appendix \ref{appa}.

\section{Price of Fairness}\label{pof}

\def\eps{\varepsilon}
\def\A{{\cal A}}

Finally, we consider the 
efficiency of allocation from the perspective of
local fairness.
We follow the approach introduced by Caragiannis {\em
  et al.}~\cite{ckkk} of studying the price of envy-freeness, 
and for this we begin with the following definitions.
Recall that for an allocation $\mathcal{A}$ into pieces
$\{A_1, A_2, \ldots, A_n\}$ for the $n$ agents, we use 
$V_i(A_i)$ to denote agent $i$'s valuation for its piece.

\begin{defi}[Optimality]
  An allocation $\mathcal{A}$, is said to be \emph{optimal} 
  if $\sum_i V_i(\mathcal{A}_i) \geq \sum_i V_i(\mathcal{B}_i)$ for
  any allocation $\mathcal{B}$. We denote this optimal allocation by
  $\mathcal{A}^*$.
\end{defi}

We define the \emph{optimal locally envy-free} (resp. \emph{optimal locally proportional})
allocations, denoted by $\mathcal{A}^{\text{LEF}^*}$ (resp. $\mathcal{A}^{\text{LP}^*}$), 
analogously by imposing the constraint that $\mathcal{A}$
and $\mathcal{B}$ be locally envy-free (resp. locally proportional)
and maximizing sum of the values across all agents. 

\begin{defi}[Price of Local Envy-Freeness, Proportionality]
Given an instance of a cake cutting problem on a graph $G$, the \emph{price of local envy-freeness} is the ratio,
$$\frac{\sum_i V_i(\mathcal{A}_i^*)}{\sum_i V_i\({\mathcal{A}_i^{\text{LEF}^*}}\)},$$
where the sum is over all agents $i \in N.$ We likewise define the price of proportionality by taking the denominator to be 
$\sum_i V_i\({\mathcal{A}_i^{\text{LP}^*}}\)$.
\end{defi}

We are measuring the degradation in efficiency when considering 
allocations that maximize the welfare in both instances under the given 
constraints.
%
To quantify the loss of efficiency, we are interested in giving a
tight lower and upper bound. More specifically, given a graph $G$ and
a fairness concept in consideration, say local envy-freeness, we seek
to find an input (i.e., a valuation profile) for which the price of
local envy-freeness is high. This corresponds to a lower bound on the
price of fairness. On the other hand, the upper bound will be given
via an argument that shows, for any valuation profile,  the price
of fairness cannot exceed that stated.  


The main result on global envy-freeness, due to Caragiannis {\em
  et al.}~\cite{ckkk} is an $\Omega(\sqrt{n})$ lower
bound on the price of (global) fairness: there exist valuation
functions for which the ratio is $\Omega(\sqrt{n})$.
(Very little is know about the upper bound for 
the price of envy-freeness: an upper
bound of $n$ is immediate, and the best known upper bound is $n -
1/2$ \cite{ckkk}.)

These existing results are for the standard model in which each agent
can envy every other agent.  Using our graph-theoretic
formulation, we can study the price of local fairness. 
As defined in Section~\ref{prelim}, this 
is  the ratio of the total welfare of
the optimal allocation to the maximum total welfare of any allocation
that is locally envy-free. 

The numerator of this 
ratio--- based on the optimal allocation ---is independent of $G$, 
while the denominator is a maximum over a set of allocations that is
constrained by $G$.
Now, if we imagine reducing the set of edges in $G$, the set of
allocations eligible for the maximum in the denominator becomes less
constrained; consequently, we would expect that the price of fairness
may become significantly smaller as $G$ becomes sparser.  Is this in
fact the case?
We show that it is not. Our main result is that the lower bound for
global envy-freeness also applies to local envy-freeness on any connected
undirected graph.
\begin{thm}
  For any connected undirected graph $G$, there exist valuation
  functions for which the price of local envy-freeness on $G$ is
  $\Omega(\sqrt{n})$.
\label{thm:local-pof}
\end{thm}

To prove this theorem, we start by adapting a set of valuation
functions that Caragiannis {\em et al}~\cite{ckkk} used in their lower
bound for global envy-freeness.  To argue about their effect on
allocations in an arbitrary graph $G$, we need to reason about the
paths connecting agents in $G$ to others with different valuation
functions.  This, in turn, requires a delicate graph-theoretic
definition and argument: we introduce a structure that we term a {\em
  $(k,\eps)$-linked partition}; we show that if $G$ contains this
structure, then we can carry out the lower bound argument in $G$; and
finally we show that every connected undirected graph contains a {\em
  $(k,\eps)$-linked partition}. 
%
\begin{definition}
For a connected graph $G = (V,E)$, a natural number $k \geq 1$, and
a real number $0 < \eps \leq 1$, 
we define a {\em $(k,\eps)$-linked partition} as follows.  
It consists of a set $L \subseteq V$ of size
$k$, and a partition of $S = V - L$ into sets $\{S_i : i \in L\}$ each
of size at least $(\eps n / k) - 1$, such that for each $j \in S_i$, there
is an $i$-$j$ path in $S \cup \{i\}$. That is, each $j \in S_i$ can
reach $i \in L$ without passing through any other nodes of $L$.
\label{def:linked-partition}
\end{definition}

We next show that if a connected undirected graph $G$ has such 
a structure, with appropriate values of $k$ and $\eps$, 
then we obtain a lower bound on the price of local
envy-freeness on $G$.

\begin{lm}
If a connected undirected graph $G$ has a 
$(k,\eps)$-linked partition with $k = \lfloor \sqrt{n} \rfloor$
and $\eps$, a constant, then 
there exist valuation functions on the nodes of $G$ for which
the price of local envy-freeness is $\Omega(\sqrt{n})$.
\label{lemma:linked-pof-lb}
\end{lm}

\begin{proof}
Suppose $G$ has a $(k,\eps)$-linked partition consisting of $L$ and
$\{S_i : i \in L\}$, where $k = \sqrt{n}$.  
Thus, each $S_i$ has size at least
$\eps \sqrt{n} - 1$.  
(We will assume for the sake of exposition that $\sqrt{n}$ is an integer,
although it is straightforward to slightly modify the argument if it is not.)

We now use a valuation function adapted from the construction
of Caragiannis et al \cite{ckkk}, who considered the price of global 
envy-freeness.
We partition the full resource to be allocated, the interval $[0,1]$,
into $\sqrt{n}$ disjoint intervals $I_1, \ldots, I_{\sqrt{n}}$.
We will give each $i \in L$ a valuation
$v_i$ that places all value on distinct interval $I_i$,
and each $j \in S$ a valuation that is uniform on
$[0,1]$.  The optimal allocation for this set of valuations has total
welfare of $\sqrt{n}$, which is achieved by giving each $i \in L$ 
the entire interval where
it places value.  

Now let us consider any envy-free allocation $\A = \{A_i : i \in V\}$. 
Let $\mu_i$ be a real number denoting the Lebesgue measure of the set
$A_i$ assigned to node $i$.
If $j \in S$, then $j$'s valuation for its set, $v_j(A_j)$ is equal to $\mu_j$.
If $i \in L$, then $i$'s valuation $v_i(A_i)$ is $\sqrt{n}$ times the
measure of $A_i \cap I_i$; hence $v_i(A_i) \leq \mu_i \sqrt{n}$.

For each $i \in L$, each $j \in S_i$ has a path $P_j$ to $i$ 
entirely through nodes of $S$;
let the nodes on this path, beginning at $i$, be
$P_j = i, j_1, j_2, \ldots, j_d = j$.
The immediate neighbor $j_1$ of $i$ on $P_j$ must satisfy
$\mu_{j_1} \geq \mu_i$, since $j_1$ is in $S$ and hence has a uniform
valuation on intervals.
For each successive $j_t$ on $P$, we must have $\mu_{j_t} = \mu_{j_{t-1}}$,
since $j_t$ and $j_{t-1}$ have the same valuation on all sets, and
the allocation is locally envy-free.
Thus, by induction we have $\mu_{j_t} \geq \mu_i$ for all $t$,
and hence $\mu_j \geq \mu_i$.

We can now derive a set of inequalities that establishes the lower bound.
First we have,
$$\sum_{j \in S_i} v_j(A_j) = \sum_{j \in S_i} \mu_j \geq \sum_{j \in S_i} \mu_i \geq \mu_i (\eps \sqrt{n} - 1).$$
Let us assume $n$ is large enough that 
$\eps \sqrt{n} - 1 \geq \eps \sqrt{n}/2$, so we have,
$$\sum_{j \in S_i} v_j(A_j) \geq \mu_i (\eps \sqrt{n} / 2).$$
Since $v_i(A_i) \leq \mu_i \sqrt{n}$ 
for $i \in L$, we have,
\begin{align}
\sum_{j \in S_i} v_j(A_j) \geq \eps v_i(A_i)/2.
\label{eq:dp}
\end{align}

Thus, the total welfare of the allocation is, 
\begin{align*}
\sum_{h \in V} v_h(A_h) 
& = \sum_{i \in L} v_i(A_i) + \sum_{j \in S} v_j(A_j) \\
&=  \sum_{i \in L} [v_i(A_i) + \sum_{j \in S_i} v_j(A_j)] \\
& \leq  \sum_{i \in L} [(2 \eps^{-1} + 1) \sum_{j \in S_i} v_j(A_j)] \\
& =  (2 \eps^{-1} + 1) \sum_{i \in L} \sum_{j \in S_i} v_j(A_j) \\
& =  (2 \eps^{-1} + 1) \sum_{j \in S} v_j(A_j) \leq  (2 \eps^{-1} + 1),
\end{align*}
where the first inequality is by~\eqref{eq:dp} and the second
since 
$$\sum_{j \in S} v_j(A_j) = \sum_{j \in S} \mu_j \leq 1,$$
because 
all agents in $S$ get disjoint intervals.
Since $(2 \eps^{-1} + 1)$ is a constant, 
while the optimal allocation has total welfare
$\sqrt{n}$, this implies an $\Omega(\sqrt{n})$ lower bound for the price
of envy-freeness on $G$.
\end{proof}

Finally, we establish that every connected undirected graph $G$ has
a $(k,\eps)$-linked partition for appropriate values of $k$ and $\eps$.
We begin by showing that it is enough to find a structure satisfying a
slightly more relaxed definition, in which the set $L$ can have more than 
$k$ elements, and we do not need to include all the nodes of $G$.
Specifically, we have the following definition:

\begin{definition}
For a connected graph $G = (V,E)$, a natural number $k \geq 1$, and
a real number $0 < \eps \leq 1$, 
we define a {\em $(k,\eps)$-linked subpartition} as follows. 
It consists of a set $L \subseteq V$ of size
$\ell \geq k$, together with disjoint subsets 
$S_1, S_2, \ldots, S_\ell \subseteq S = V - L$, each of size
of size at least $(\eps n / k) - 1$, such that for each $j \in S_i$, there
is an $i$-$j$ path in $S \cup \{i\}$. 
\label{def:linked-subpartition}
\end{definition}

The following lemma says that it is sufficient to find a
$(k,\eps)$-linked subpartition.

\begin{lm}
If a connected undirected graph $G$ contains a
$(k,\eps)$-linked subpartition, then it contains a 
$(k,\eps)$-linked partition.
\label{lemma:subpartition}
\end{lm}

\begin{proof}
We start with a $(k,\eps)$-linked subpartition of $G$, with disjoint sets
$L$ of size $\ell \geq k$, and $S_1, S_2, \ldots, S_\ell \subseteq S = V - L$.
First, for every node $v \not\in L \cup S$, we assign it to a subset $S_i$
as follows: we find the shortest path from
$v$ to any node in $L$; suppose it is to $i \in L$.
We add $v$ to $S_i$.  Note that this preserves the property that
all $S_i$ are disjoint, and $v$ has a path to $i$ that does not meet
any other node of $L$, since if $h \in L$ were to lie on this path,
it would be closer to $v$ than $i$ is.

At this point, every node of $G$ belongs to $L \cup S$.
We now must remove nodes from $L$ to reduce its size to exactly $k$
while preserving the properties of a $(k,\eps)$-linked partition.
To do this, we choose a node $i \in L$ arbitrary, remove 
$i$ from $L$, and remove the set $S_i$ from the collection of subsets.
We then assign each node in $S_i \cup \{i\}$ to an existing subset
$S_h$ exactly as in the previous paragraph.
After this process, the size of $L$ has been reduced by $1$, and we
still have a partition of $V - L$ into subsets $S_i$ with the desired
properties.
Continuing in this way, we can reduce the size of $L$ to exactly $k$,
at which point we have a $(k,\eps)$-linked partition.
\end{proof}

Finally, we prove the following graph-theoretic result, which together
with Lemma \ref{lemma:linked-pof-lb}
establishes Theorem \ref{thm:local-pof}.

\begin{thm}
For every $k \geq 2$ and with $\eps = 1/2$, every 
connected undirected graph has a $(k,\eps)$-linked subpartition.
\label{thm:subpartition-exists}
\end{thm}

\begin{proof}
It is enough to find the required structure on a spanning
tree $T$ of $G$, since if the paths required by the definition exist in $T$,
then they also exist in $G$.
Thus, it is sufficient to prove the result for an arbitrary tree $T$.

We root $T$ at an arbitrary node, and let $X$ be the set of leaves of $T$.
If $|X| \geq k$, then we can choose any $k$ leaves of $T$ and partition
the remaining nodes of $T$ arbitrarily into sets of size
$(n-k)/k$ to satisfy the definition.
Otherwise, $|X| < k$.  In this case, we begin by including all nodes
of $X$ in $L$.

Now, we process the nodes of $T$, working upward from the leaves,
so that when we get to a node $v$ in $T$, we have already processed
all descendents of $v$.
Each node is processed once, and at that point we decide whether
to add it to $L$, and if not which set $S_i$ to place it in,
given the current set $L$.
For a node $v$, we say that $w$ is
{\em downward-reachable} from $v$ if $w$ is a descendent of $v$, and if
the $v$-$w$ path in $T$ does not
contain any internal nodes belonging to $L$.

When we process a node $v$, we do one of two things:
\begin{itemize}
\item[(i)] 
We {\em label} $v$ with the name of a node in $L$ that is downward-reachable
from $v$; or
\item[(ii)] We place $v$ in $L$.
\end{itemize}

Let $b = (\eps n / k) - 1$.
We perform action (i) if there is any $w \in L$ that is downward-reachable
from $v$, such that there are not yet $b$ nodes labeled with $w$.  
In this case, we label $v$ arbitrarily with one such $w$.
Since $v$ and all its descendents are now processed, $v$ will continue
to have a path to $w$ that does not pass through any other nodes of $L$.
Otherwise suppose there is no such $w$; that is,
all $w \in L$ that are downward-reachable from $v$ have $b$ nodes
labeled with $w$.
In this case, we perform action (ii).  Note that at this point, 
every $w \in L$ that is a descendent of $v$ has a set $S_w$ of
exactly $b$ nodes, and these nodes can all reach $w$ without passing
through any other node of $L$.

Our procedure comes to an end when we process the root node $v^*$.
There are three cases to consider, the first two of which are straightforward.

First, 
if we place $v^*$ into $L$, then $T - v^*$ is partitioned into $L$ and sets 
$\{S_w : w \in L\}$ such that $|S_w| = b$ for each $w$.
Thus, if we remove $v^*$ from $L$, we have a 
$(k,\eps)$-linked subpartition, since the sets $S_w$ are disjoint and of
size at least $b$, and 
$$|L| = (n-1) / ((\eps n / k) - 1) \geq (n / (\eps n / k)) = k / \eps > k.$$

Otherwise, $v^*$ is labeled with some downward-reachable $u \in L$.
Our second case, which is also straightforward, is that 
after this labeling of the root, all sets $S_w$ for $w \in L$ have
size exactly $b$, then we have a $(k,\eps)$-linked subpartition.

If not, then we are in the third case: 
$v^*$ is labeled with some downward-reachable $u \in L$, and
after this labeling there still exist downward-reachable 
nodes $w$ that we have placed in $L$
that do not have associated sets $S_w$ of size $b$.
We therefore need to prune our set $L$ to a smaller set that has
$|S_w| \geq b$ for each $w \in L$.
The goal is to show that the smaller $L$ we end up with still has
enough elements; if that holds, then we have a
$(k,\eps)$-linked subpartition.

To show this, we proceed as follows.
We say that $w \in L$ is {\em active} if $|S_w| < b$.  We first
observe that any active $w$ must be downward-reachable from the root
$v^*$.  Indeed, if $w$ is active and not downward-reachable from the
root, then there is a $v \in L$ such that $w$ is a descendent of $v$.
But in the step when we placed $v$ in $L$, it was not possible to
label $v$ with $w$, and hence we must have had $|S_w| = b$ at that
point.

Next we claim that there are $< k$ active $w \in L$.
To prove this, for each active $w$, we associate $w$ with a leaf that is a
descendent of $w$.  (This can be $w$ itself if $w$ is a leaf.)
Observe that the same leaf $x$ cannot be associated with two distinct
active $w, w'$, for then on the path from $v^*$ to $x$, one of 
$w$ or $w'$ would be closer to $v^*$, and the other would not be
downward-reachable from $v^*$.
Given that we can associate a distinct leaf to each active $w$,
and there are $< k$ leaves, there are $< k$ active $w \in L$.

We say that a node $w \in L$ is {\em inactive} if it
is not active; that is, if $|S_w| = b$.
Let $L_0$ be the inactive nodes of $L$ and $L_1$ be the active nodes of $L$.
We have
$$|L_0| + \sum_{w \in L_0} |S_w| + |L_1| + \sum_{w \in L_1} |S_w| = n.$$
We know that $|L_1| < k$ and $|S_w| < b$ for each $w \in L_1$;
hence, using the fact that $\eps = 1/2$, we have
$$|L_1| + \sum_{w \in L_1} |S_w| < k + k b = k (b + 1) = k n / (2k) = n/2.$$
It follows that 
$|L_0| + \sum_{w \in L_0} |S_w| > n/2$.
But since $|S_w| = b$ for each $w \in L_0$, we have
$$n/2 < |L_0| + \sum_{w \in L_0} |S_w| = |L_0|(1 + b) = |L_0| n / (2k),$$
from which it follows that $|L_0| > k$.
We now conclude the construction by declaring $L$ to be $L_0$; 
since the sets $S_w$ for $w \in L_0$ are all pairwise disjoint and
each has size at least $b$, we have the desired
$(k,\eps)$-linked subpartition.
\end{proof}

\omt{
We conclude this section by noting that as in the work of
Caragiannis et al \cite{ckkk}, the lower bound is the main result.
Like in their work, we can only provide a mild improvement on
the immediate upper bound of $n$.  
In particular, we show how to generalize their basic upper bound 
of $n - 1/2$ to incorporate the number of connected components of $G$.

\begin{thm}
Given a graph $G$ with $r$ components, the price of local-envy-freeness is at most $n - r/2$. 
\end{thm}

\begin{proof}
Given an instace $G$, consider the corresponding optimal allocation.
If this allocation is locally envy-free, then the price of fairness is
$1$. However, if it is not envy-free, then there must be at least one
agent in each component that envies the allocation of another agent.
Therefore, the value for their allocation in $\mathcal{A}^*$ is at
most one half. Thus, the denominator is at most $n-1/2 \cdot
r$. For the numerator, we know that there exists a globally envy-free
allocation such that the sum of the utilities of the agents is at
least $1$ since a globally envy-free allocation is also globally
proportional. Therefore, the optimal locally envy-free allocation must
have welfare at least $1$, giving us the desired upper bound.
\end{proof}


}

\section{Conclusion and Future Work}\label{conc}

We have introduced a new line of inquiry for the envy-free and
proportional cake cutting problems by considering local notions of
fairness. We show interesting relations between these local fairness
concepts and their global analogues.  Besides introducing this new
model, our main contribution has been to fully classify the class of
graphs for which there is an oblivious single-cutter protocol for
computing locally envy-free allocations.  Furthermore, we quantify the
degredation in welfare resulting from adding the local envy-freeness
constraint on the allocations; in particular, we show that the known
$\Omega(\sqrt{n})$ lower-bound for the (global) price of envy-freeness
continues to hold even for sparse graphs.

It is of interest to give 
efficient protocols for computing locally envy-free allocations on
rich classes of graphs without the single-cutter constraint.  Since
local envy-freeness is a stronger condition than local
proportionality, the same problem can also be considered for locally
proportional allocations. Finally, whether there is a similar lower bound of 
$\Omega(\sqrt{n})$ for the price of local proportionality is an open question. Currently,
the upper bound for both local fairness concepts is the loose $n -
1/2$ bound, and giving tighter bounds is another direction.

\section*{Acknowledgements}
The first author would like to thank Felix Fischer for the introduction 
to the cake cutting problem. We would also like to thank Rahmtin Rotabi 
for insights for the proof of Theorem 4.4 as well as Barabra Grosz, 
Radhika Nagpal, Ariel Procaccia, Ofra Amir, and Adam Brauer for 
helpful discussions and references.

\newpage

\appendix

\section{Taking Advantage of Irrevocable Advantage}\label{appa}

Given a partial envy-free allocation, an agent $i$ is said to 
{\em dominate }
an agent $j$, if $i$ remains envy-free of $j$ even if the entire residue
(the remaining subinterval of the cake) is allocated to $j$. We can thus define:

\begin{defi}
A \emph{domination graph} on $n$ is a graph where $V$ is the set of agents and there is a 
directed edge $(i, j)$ if $i$ has irrevocable advantage over $j$. 
\end{defi}

Constraints on the number of agents each agent must dominate at a
certain stage in the protocol can be used to extend partial envy-free
allocations to complete ones. One salient example is the
Aziz-Mackenzie protocol for $K_4$, where they obtain a partial
envy-free allocation using what they call the
{\em Core Protocol} and the
{\em Permutation Protocol} such that each agent is guaranteed to dominate at
least two other agents.  They then use the Post-Double Domination Protocol to 
extend this to a
complete envy-free allocation.  We generalize this protocol to any
$n$. That is, given a partial envy-free allocation such that each
agent dominates $n-2$ other agents, we can apply Protocol 1 to extend
the allocation to a complete allocation. This provides an alternate proof to a
recent paper by Segal-Halevi et al.~\cite{sha}. We 
first define a special class of graphs.
\begin{defi}
A \emph{pseudoforest} is a graph where each vertex has at most one outgoing edge. 
\end{defi}

Each component of a pseudoforest is a subgraph of a cone of a DAG;
since each node has at most one outgoing edge, there are at most $n$
edges. If there are fewer than $n$, then it is a DAG. If there are
exactly $n$, then there exists a cycle. Find this cycle and remove
an edge $e = (u, v)$ from the cycle. The resulting graph will
be a DAG, and we can therefore Protocol 1 by setting $u = c$. This makes the
protocol key to extending partial envy-free allocations to complete
ones under particular domination criteria.
\begin{lm}
We can extend a partial globally envy-free allocation in which each agent dominates at 
least $n-2$ other agents to a complete, envy-free allocation. 
\end{lm}

\begin{proof}
Suppose we have a partial globally envy-free allocation $(P_1, P_2, \cdots, P_n)$, with 
residue $R$. If each agent dominates at least $n-2$ other agents, then the complement 
of the domination graph, denoted by $G^c$, is a pseudoforest. We apply Protocol 1 on 
$G^c$ using residue $R$ and denote this allocation by $(R_1, R_2, \cdots, R_n)$. We 
want to show that $(P_1 \cup R_1, P_2 \cup R_2, \cdots, P_n \cup R_n)$ is a globally 
envy-free allocation. Suppose it is not. Then, there exists $i, j$ such that 
$V_i(P_i \cup R_i) < V_i(P_j \cup R_j)$, but this is only possible if either 
$V_i(P_i ) < V_i(P_j)$ or $V_i(R_i) < V_i(R_j)$.
\end{proof}

The assumption that each agent dominates at least $n-2$ agents is
necessary. In particular, suppose that there exists one agent that
dominates only $n-3$ other agents, such as in Example~\ref{ex:2}.
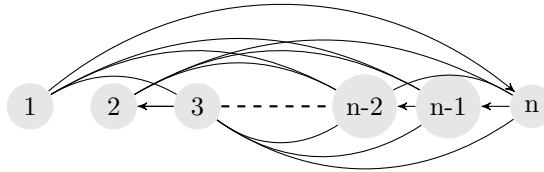
\begin{figure}[ht]
\centering
\begin{tikzpicture}
  [scale=1.1,auto=left,every node/.style={circle,fill=gray!20}]
  \node (n1) at (0, 0) {1};
  \node (n2) at (1, 0)  {2};
  \node (n3) at (2, 0)  {3};
  \node (n4) at (4, 0) {n-2};
  \node (n5) at (5, 0) {n-1}; 
  \node (n6) at (6, 0) {n};

\draw (n1) edge[bend left=30] (n3);
\draw (n1) edge[bend left=30] (n4);
\draw (n1) edge[bend left=30] (n5);
\path [->] (n1) edge[bend left=40] (n6); 

\draw (n2) edge[bend left=30] (n4);
\draw (n2) edge[bend left=30] (n5);
\draw (n2) edge[bend left=30] (n6);

\draw (n3) edge[bend right=35] (n4);
\draw (n3) edge[bend right=35] (n5);
\draw (n3) edge[bend right=35] (n6);
\draw (n3) edge[bend right=0] (n2);
\path [->] (n3) edge[bend right=0] (n2);

\draw[thick, dashed] (n3) -- (n4); 

\draw (n4) edge[bend left=30] (n6);

\path [->] (n5) edge[bend right=0] (n4); 
\path [->] (n6) edge[bend right=0] (n5); 

\end{tikzpicture}
\caption{Counterexample to the extension lemma.\label{fig:counter}}
\end{figure}

\begin{eg}
\label{ex:2}
Suppose that each agent $i \neq 2$ dominates every other agent but agents 
$i+1 \mod n$ and that agent $2$ dominates agents $\{4, 5, \cdots, n\}$. The 
domination graph is given in Figure~\ref{fig:counter}.
The complement of the domination graph is the cycle graph 
$(1, 2, 3, \cdots, n, 1)$ plus the edge $(2, 1)$. It therefore consists 
of more than one simple cycle, and hence a direct application of 
Protocol 1 to the complement of the domination graph will not 
extend a partial allocation to a complete allocation. 
\end{eg}

\end{document}